\documentclass[runningheads]{llncs}
\usepackage[T1]{fontenc}
\DeclareFontShape{T1}{cmr}{m}{scit}{ <-> ssub * cmr/m/sc }{}
\usepackage{graphicx}
\usepackage{amsmath}
\usepackage{amssymb}
\usepackage{booktabs}
\usepackage{enumitem}
\usepackage{tikz}
\usepackage{tabularx}
\usetikzlibrary{arrows.meta,positioning,shapes.geometric}

\title{Probabilistic Model Checking of Autoregressive Neural Sequence Models}
\titlerunning{Probabilistic Model Checking of Autoregressive Neural Models}

\author{Helge Spieker\orcidID{0000-0003-2494-4279} \and
Dennis Gross\orcidID{0009-0001-5734-0538} \and
Arnaud Gotlieb\orcidID{0000-0002-8980-7585}}

\authorrunning{H. Spieker et al.}

\institute{Simula Research Laboratory, Oslo, Norway\\
Corresponding author: \email{helge@simula.no}}

\begin{document}

\maketitle

\begin{abstract}
Test-set accuracy is silent on two issues that matter when deploying autoregressive neural sequence models: how much probability mass the system under test (SUT) places on constraint-violating alternatives that are reachable under sampling and what fraction of the input population satisfies a domain requirement. 
We answer both with probabilistic model checking. The pipeline extracts a discrete-time Markov chain (DTMC) from the SUT's token-by-token generation, verifies formal PCTL specifications with the PRISM model checker, and aggregates the per-input verdicts into a coverage curve over the input space. A soundness theorem establishes the DTMC as an under-approximation, so every verdict yields a certified interval on the SUT's true reachability probability. The coverage built from those verdicts is, therefore, conservative by construction. A counterexample-guided abstraction refinement (CEGAR) loop adaptively tightens the interval, and a maximum-likelihood algorithm extracts the most probable falsifying trace. 
Two case studies exercise the pipeline.
On a GPT-2 computer-aided process-planning (CAPP) model with 100\% test accuracy, the pipeline quantifies the probability mass greedy decoding hides, but that is reachable with sampling; and identifies the smallest training fraction at which an ordering requirement holds population-wide, neither of which test accuracy can report.
We then verify the SMILES molecular generator with a $50\times$ larger vocabulary. The only change is an external chemical-validity oracle, and the pipeline identifies the gap between structural completeness and chemical validity.

\keywords{probabilistic model checking \and DTMC abstraction \and PCTL
  \and autoregressive
  transformers \and neural network verification}
\end{abstract}

\section{Introduction}\label{sec:intro}

Autoregressive transformer models, the technology behind large language models, now appear in engineering pipelines, where generated sequences affect physical artifacts, regulatory outcomes, and downstream automation. 
However, their reliability is far from guaranteed~\cite{valmeekam2023planning,ji2023hallucination}.
Yet their evaluation still rests almost entirely on test-set accuracy: a single number reporting how often the output matches the held-out reference~\cite{liang2023holistic}.
Three questions matter for deployment that no test set can answer on its own.
First, how much probability mass does the model place on low-confidence alternatives that greedy decoding never selects, but a sampled decoder eventually will~\cite{holtzman2020curious}?
Second, do generated sequences obey domain constraints that lie outside the training signal and so never enter an accuracy metric, because it does not report failure causes or because these failures are not part of the ground-truth dataset in the first place?
Third, what fraction of the entire input population satisfies a given criterion?
Are domain violations corner cases or potentially frequent occurrences?

As an example, consider the problem of computer-aided process planning (CAPP).
High-level CAPP is the task of determining appropriate sequences of manufacturing processes (e.g., milling, drilling, grinding) to produce a given part. 
A part can usually be made in several ways, so the task is naturally a variable-length sequence-prediction problem.
Stathatos et al.~\cite{STATHATOS2026103233} address this problem with a domain-specific GPT-2 transformer~\cite{radford2019language} that maps a discrete part description to an ordered sequence of operations.
The model reports 100\% sequence accuracy on the held-out test set, a number that suggests it is ready to deploy. 
Our pipeline shows otherwise: even the smallest checkpoint with 100\% accuracy (15\% of the data) still routes over 9\% of its probability mass to completions greedy decoding never selects, and on data-starved checkpoints the most probable hidden trace violates manufacturing phase ordering.
Neither failure mode is visible to accuracy, because both reference probability mass or a domain constraint the ground-truth sequence never exercises.
This hidden mass is a risk, but it is also a lever. 

A deployed model must commit to a decoding rule: \emph{greedy} decoding always emits the single most likely next operation, whereas a \emph{sampled} decoder draws from the model's distribution and so can produce different plans on repeated runs. 
This can potentially improve performance, because the best next token might not always have the highest probability; especially in settings with limited data.
Because our pipeline reasons about the whole distribution rather than one trajectory, it predicts (without re-running the model) how often a sampled decoder, paired with the verifier as a filter, recovers a correct plan greedy decoding would miss~\cite{cobbe2021verifiers}.
On the most data-starved checkpoint a slightly sharpened sampler does exactly that, turning the risk into a usable safety margin.

In our approach, we treat the autoregressive transformer as the \emph{system under test} (SUT) and apply a three-stage probabilistic-model-checking pipeline.
It extracts a \emph{discrete-time Markov chain (DTMC)} by unrolling the token-by-token generation, verifies Probabilistic Computation Tree Logic (PCTL) specifications against it with the PRISM model checker~\cite{Kwiatkowska11}, and aggregates the per-input verdicts into a \emph{coverage curve} over the input population.
Because the DTMC is a sound under-approximation, each input's set of PRISM verdicts yields a certified \emph{interval} on the SUT's reachability probability, not a point estimate, and the population fractions built from these intervals are themselves conservative.
A CEGAR-style refinement loop~\cite{Clarke00cegar} tightens the interval adaptively, and a maximum-likelihood counterexample algorithm extracts the most probable falsifying trace from any violated specification.
The result is a \emph{population coverage statement}: the fraction of inputs for which property $\phi$ holds with probability at least $\theta$.

Overall, this paper makes the following contributions:
\begin{enumerate}
  \item A \emph{probabilistic model verification methodology} for autoregressive transformers built on three artifacts: a DTMC abstraction of the SUT (Sec.~\ref{sec:extraction}), PCTL specifications including domain-specific properties (Sec.~\ref{sec:ordering}), and population coverage over the input space (Sec.~\ref{sec:pac-method}).
  \item A \emph{soundness theorem} (Theorem~\ref{thm:soundness}) that yields a certified interval on the SUT, tightened adaptively by a \emph{CEGAR-style refinement loop} (Sec.~\ref{sec:cegar}).
  \item A \emph{maximum likelihood counterexample algorithm} (Sec.~\ref{sec:witness-method}) extracting the most probable falsifying trace as an inspectable artifact. 
  \item Experiments on two systems: \textit{CAPP} (Sec.~\ref{sec:capp-eval}) exposing failure modes invisible to accuracy, and \textit{SMILES} molecule generation (Sec.~\ref{sec:smiles}) with a substantially larger vocabulary.
\end{enumerate}

\section{Related Work}\label{sec:related}

\emph{Model checking and abstraction of neural systems.}
A close line of work applies probabilistic model checking to neural controllers:
MOSAIC \cite{bacci2020probabilistic} abstracts the closed-loop system as an MDP (or interval MDP) and verifies PCTL properties with PRISM.
Similarly, Gross et al.\ extract DTMCs from RL policies~\cite{DBLP:conf/sac/GrossS25,DBLP:conf/esann/GrossS24} or bounded LLM outputs~\cite{DBLP:conf/ictai/GrossSG25} for safety verification.
Our approach differs in its target: the joint distribution of an autoregressive generation process (a DTMC) rather than a policy under environment stochasticity (an MDP), and it adds an aggregation layer over the \emph{input population}.
Weiss et al.~\cite{Weiss2018} extract finite automata from recurrent networks via L*-style equivalence queries. Here, our extraction targets a probabilistic surrogate by direct forward unrolling under threshold
pruning, with the soundness gap explicitly bounded by Theorem~\ref{thm:soundness}.

\emph{Statistical and deterministic verification.}
UPPAAL SMC~\cite{David15} and Plasma~\cite{Legay10} estimate
probabilities by simulation; we compute exact per-input reachability
via PRISM's backward induction and use the aggregation layer only to combine those exact verdicts into a population.
Viewed differently, DTMC extraction is symbolic execution of a probabilistic program~\cite{GordonHNR14}, threshold pruning is an abstract interpretation~\cite{CousotCousot77}, and the refinement loop adapts CEGAR~\cite{Clarke00cegar} to the DTMC setting.
Deterministic verifiers, e.g., Reluplex/Marabou~\cite{Katz17,Katz19}, ERAN~\cite{Singh19}, or $\alpha$-$\beta$-CROWN~\cite{Wang21}, verify a single forward pass of a feed-forward network; our target is the
\emph{generative process} over multiple decoding steps, which requires
a probabilistic formalism.

\emph{Population-level guaranties and neural-network testing.}
Conformal
prediction~\cite{Quach2024,MohriHashimoto2024,Angelopoulos23} provides
population-level claims over black-box output distributions. 
It is model-agnostic and does not reference internals, whereas our pipeline accesses the SUT's token probabilities and enables domain-specific PCTL specifications such as phase ordering.
Coverage-guided testing (DeepXplore~\cite{Pei2017}, DeepGauge~\cite{Ma2018}) targets activation-space coverage of classifiers; our coverage curve is analogous to a coverage criterion, but with soundness guaranties.

\section{Verification Pipeline}\label{sec:pipeline}

\subsection{Preliminaries}\label{sec:prelim}

An autoregressive model over a finite vocabulary $V$ factorizes a sequence
distribution left to right, $P(x) = \prod_{t=1}^{n} P(x_t \mid x_{<t})$, each
conditional a softmax over the logits; a temperature $T$ rescales the logits by
$1/T$ before the softmax ($T \to 0$ recovers greedy decoding). This
factorisation maps onto a \emph{discrete-time Markov chain} (DTMC)
$D = (S, s_0, \mathbf{P}, L)$ whose states are token prefixes, whose
transitions carry the conditionals, and whose labeling $L$ marks states with
atomic propositions, because each transition strictly extends the prefix the
chain is acyclic, so reachability probabilities follow from a single
backward-induction pass with no iterative error. We specify properties in
\emph{PCTL}~\cite{DBLP:journals/fac/HanssonJ94} using the quantitative
reachability form $\mathtt{P}_{=?}[\mathtt{F}\,\textit{label}]$, the exact
probability of eventually reaching a \textit{label}-state, computed by
PRISM~\cite{Kwiatkowska11}. An aggregation layer (Sec.~\ref{sec:pac-method})
summarises these per-input verdicts to a coverage curve over the input population.

\subsection{Pipeline Overview}
The pipeline has three independently operable stages (Figure~\ref{fig:pipeline}): 
\emph{DTMC abstraction} of the SUT, \emph{specification verification} by PRISM model checking, and \emph{coverage aggregation} that combines per-input verdicts into a population-level coverage curve.
Each stage reads from and writes to files, so any stage can be re-run or replaced independently.
Two auxiliary components operate alongside: a CEGAR-style refinement loop (Sec.~\ref{sec:cegar}) that re-expands the highest-impact pruned transitions, and a counterexample extractor (Sec.~\ref{sec:witness-method}) that produces an inspectable falsifying trace for any violated specification.
The following presentation is agnostic of the specific SUT. 
The two case studies in Sections~\ref{sec:capp-eval} and~\ref{sec:smiles} instantiate the same pipeline on very different SUTs and specifications.

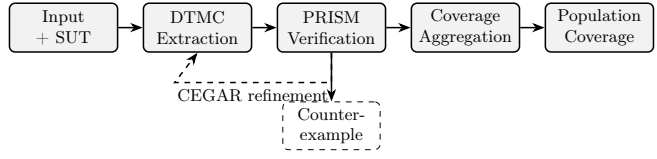
\begin{figure}[t]
\centering
\resizebox{0.7\textwidth}{!}{%
\begin{tikzpicture}[
  node distance=0.5cm and 0.45cm,
  mainbox/.style={draw, rounded corners=3pt, minimum width=2.0cm,
                  minimum height=0.9cm, align=center,
                  fill=gray!10, thick},
  outbox/.style={draw, rounded corners=3pt, minimum width=1.8cm,
                 minimum height=0.7cm, align=center,
                 fill=white, dashed},
  arr/.style={-Stealth, thick},
  darr/.style={-Stealth, thick, dashed},
]
\node[mainbox] (sut)   {Input\\+ SUT};
\node[mainbox, right=of sut]   (dtmc)   {DTMC\\Extraction};
\node[mainbox, right=of dtmc]  (prism)  {PRISM\\Verification};
\node[mainbox, right=of prism] (pac)    {Coverage\\Aggregation};
\node[mainbox, right=of pac]   (cert)   {Population\\Coverage};

\draw[arr] (sut)   -- (dtmc);
\draw[arr] (dtmc)  -- (prism);
\draw[arr] (prism) -- (pac);
\draw[arr] (pac)   -- (cert);

\draw[darr]
  (prism.south) -- ++(0,-0.55cm)
  -- node[below] {\small CEGAR refinement}
  ++(-2.9cm,0) -- (dtmc.south);

\node[outbox, below=0.9cm of prism] (witness) {Counter-\\example};
\draw[arr] (prism.south) -- (witness.north);
\end{tikzpicture}%
}
\caption{Pipeline overview. A per-input DTMC abstraction is verified by
  PRISM against PCTL specifications; verdicts aggregate into a population
  coverage curve. The CEGAR loop (dashed) re-expands the
  highest-impact pruned transitions; counterexamples are extracted on
  demand from violated specifications.}
\label{fig:pipeline}
\end{figure}

\subsection{DTMC Extraction}\label{sec:extraction}

\paragraph{Tree expansion.}

Given an input, the extractor builds the DTMC by breadth-first expansion.
At each non-terminal state $s$, tokens with conditional probability $\ge \tau$
receive explicit successor states; the remaining mass is routed,
unrescaled, to an absorbing \textsc{low\_prob} state.
Paths whose accumulated probability falls below $\rho$ are also redirected to
\textsc{low\_prob}.
A hard depth limit $d_{\max}$ guaranties the termination. 
Paths that reach it without terminating are redirected to an absorbing \textsc{truncated} state.

Unless stated otherwise, the extractor operates at $T = 1.0$ (the model's native distribution).
Since any fixed $T$ yields a categorical distribution at each step, the extracted object is again a DTMC and Theorem~\ref{thm:soundness} applies unchanged.

\paragraph{Structural validity pruning.}
At each expansion, the extractor checks structural well-formedness
against a SUT-specific grammar of admissible prefixes; for CAPP this
forbids padding before \texttt{<EOS>}, restricts \texttt{<EOS>} to
terminal position, and requires each \texttt{<EOC>} to be preceded by
at least one process token.
Violating prefixes go to an absorbing \textsc{invalid} state.

\paragraph{Critical-state annotation.}
States with top-1/top-2 probability gap $g(s) < 10\%$ are flagged as
\emph{critical}; the flag is exported as a Boolean PRISM label and
queried via \texttt{P=?[F critical]}.
Defaults are per-token threshold $\tau = 0.5\%$, cumulative floor $\rho = 10^{-4}$, and depth limit $d_{\max} = 20$.

\subsection{PRISM Export and PCTL Specifications}\label{sec:ordering}

The pipeline exports the DTMC to the PRISM modelling language, encoding
the four absorbing sinks (\textsc{success}, \textsc{low\_prob},
\textsc{invalid}, \textsc{truncated}) and the critical-state flag as
Boolean labels.
Beyond these generic labels, the extractor accepts domain-specific
labels supplied per case study, a phase-ordering predicate on
\textsc{success} terminals for CAPP (Sec.~\ref{sec:capp-sut}) or a
chemical-validity check for SMILES (Sec.~\ref{sec:smiles}).
Two invariants hold by construction for generic labels:
$P(\textsc{success}) + P(\textsc{low\_prob}) + P(\textsc{invalid}) +
P(\textsc{truncated}) = 1$, and any domain-specification label that
partitions \textsc{success} terminals into \textsc{pass}/\textsc{fail}
satisfies $P(\textsc{pass}) + P(\textsc{fail}) = P(\textsc{success})$.
Table~\ref{tab:properties} lists the PCTL queries used in the CAPP
study; the SMILES study substitutes a chemical-validity label for the
phase-ordering pair.

\begin{table}[t]
  \centering
  \caption{PCTL specifications verified per input (CAPP instantiation).}
  \label{tab:properties}
  \small
  \begin{tabularx}{\textwidth}{XXX}
    \toprule
    Query & Meaning & Category \\
    \midrule
    \texttt{P=?[F success]}    & valid completion         & Generic \\
    \texttt{P=?[F low\_prob]}  & mass to sink             & Generic \\
    \texttt{P=?[F invalid]}    & structural violation     & Generic \\
    \texttt{P=?[F truncated]}  & depth-limit reached      & Generic \\
    \texttt{P=?[F correct]}    & ground-truth sequence    & Correctness \\
    \texttt{P=?[F ordered]}    & $P{\to}S^*{\to}F^*$ holds & Domain (CAPP) \\
    \texttt{P=?[F misordered]} & $P{\to}S^*{\to}F^*$ fails & Domain (CAPP) \\
    \texttt{P=?[F critical]}   & critical state visited   & Confidence \\
    \bottomrule
  \end{tabularx}
\end{table}

\subsection{Population Coverage}\label{sec:pac-method}

Fix a threshold $\theta \in [0,1]$ and let $\hat{\mu}(\theta) = \tfrac{1}{N}\sum_{i=1}^{N} \mathbf{1}[P_D(\phi_i) \ge \theta]$ be the fraction of inputs of the $N$ test whose verified probability of satisfying $\phi$ satisfies $\theta$.
Sweeping $\theta$ over $[0,1]$ trace the \emph{coverage curve} $\hat\mu(\theta)$, a population-level summary of how strongly the specification holds in the input space.
Sec.~\ref{sec:soundness} shows that this curve is conservatively transferred from the abstraction to the SUT itself.

The result is a \emph{population coverage statement} of the form ``$\hat\mu(\theta)$ of all inputs satisfy $\phi$ at threshold $\theta$.''
The population is \emph{covered at level $c$} for $(\phi,\theta)$ if $\hat\mu(\theta) \ge c$.
The trivial case $c=0$ asks only that the curve rule out total failure; an operational deployment criterion fixes $c$ to an application-specific floor.
Unless stated otherwise we report the coverage value $\hat\mu(\theta)$ and flag the deployment level separately.
By Lemma~\ref{lem:conservative} the reported curve already under-estimates the SUT's true coverage, so we report it directly.

\subsection{Soundness and Certification Guarantees}\label{sec:soundness}

The guarantee rests on a single property of the DTMC: it is a sound \emph{under-approximation} of the model, so every reachability probability the pipeline reports is a certified bound on the true value, conservative by construction.
Theorem~\ref{thm:soundness} bounds the abstraction gap; two corollaries discharge a term and record the mass-conservation identity; Lemma~\ref{lem:conservative} carries the per-input bound through the population layer of Sec.~\ref{sec:pac-method}.

\begin{theorem}[Sound under-approximation]\label{thm:soundness}
Let $M$ be the autoregressive model and $D$ the DTMC from
Sec.~\ref{sec:extraction}.
For any reachability property $\phi$ that holds only at \textsc{the success} terminals:
\begin{align*}
  P_D(\phi) \le P_M(\phi) \le
  P_D(\phi) &+ P_D(\textsc{low\_prob}) \\
            &+ P_D(\textsc{invalid}) + P_D(\textsc{truncated}).
\end{align*}
\end{theorem}

\begin{proof}[Sketch]
By construction, the four sinks are absorbing and, together with the \textsc{success} terminals, partition every path of $D$: each path is classified exactly once, at the step it is either completed or diverted.
Every explicit path in $D$ maps to a model path with identical probability (branched tokens carry the original softmax mass; remaining mass is rerouted intact to the appropriate sink), giving the lower bound.
For the upper bound, any model path reaching a $\phi$-satisfying terminal either survives all checks (counted by $P_D(\phi)$) or is diverted at some step into a sink.
Because diversion is absorbing, a diverted path contributes no further mass to $\phi$ within $D$ and its fate $M$ is unknown.
Attributing all diverted mass adversarially to $\phi$ gives the bound.
\end{proof}

The two sides of the bound have a simple reading. 
The lower bound holds because every explicit path that the extractor keeps is a real model path with its original probability. 
The upper bound holds because the only mass that the abstraction loses is the mass it diverted into a sink, and the worst case is that all of it would have satisfied $\phi$. 
The interval between is exactly the diverted mass.

One of those sink terms can be eliminated outright.
The structural validity grammar of Sec.~\ref{sec:extraction} is \emph{prefix-closed}. 
Once a prefix is rejected, every extension of it is rejected, too.
A path diverted to \textsc{invalid} can then never reach a $\phi$-satisfying terminal, so its mass does not belong in an upper bound on $P(\phi)$.

\begin{corollary}[Tightened upper bound]\label{cor:invalid}
Under the prefix-closed grammar of Sec.~\ref{sec:extraction},
$P_M(\phi) \le P_D(\phi) + P_D(\textsc{low\_prob}) +
P_D(\textsc{truncated})$.
\end{corollary}

\begin{corollary}[Mass conservation]\label{cor:mass}
For every per-input DTMC, the masses of the four sinks
(\textsc{success}, \textsc{low\_prob}, \textsc{invalid},
\textsc{truncated}) sum to $1$.
\end{corollary}

Corollary~\ref{cor:mass} is what makes the four sink masses a valid decomposition rather than four loosely related quantities.
They partition the unit probability, so any one of them is determined by the other three.
The evaluation uses this as a consistency check and reports the \textsc{success} lower bound alongside the certified upper bound of Corollary~\ref{cor:invalid} throughout.

It remains to show that the bound survives aggregation.
The aggregation layer counts all inputs whose verified probability clears a threshold $\theta$.
If each per-input number is a lower bound, the count built from them should itself be a lower bound on the true population fraction.

\begin{lemma}[Conservative coverage]\label{lem:conservative}
Let $\hat\mu(\theta)$ be the coverage computed from the DTMC lower bounds
$P_D(\phi_i)$, and let $\hat\mu_M(\theta)$ be the same fraction computed
from the true model probabilities $P_M(\phi_i)$ on the same inputs.
Then $\hat\mu(\theta) \le \hat\mu_M(\theta)$ for every $\theta$.
\end{lemma}

\begin{proof}
By Theorem~\ref{thm:soundness}, $P_D(\phi_i) \le P_M(\phi_i)$ for every
input $i$, so $\mathbf{1}[P_D(\phi_i) \ge \theta] \le
\mathbf{1}[P_M(\phi_i) \ge \theta]$ pointwise; summing over the test set
gives $\hat\mu(\theta) \le \hat\mu_M(\theta)$.
\end{proof}

Pessimism in the abstraction propagates as pessimism in the coverage curve: abstraction error can only make the reported coverage more conservative, yet never unsound.
A tighter $\tau$ or a CEGAR pass can only \emph{raise} the certified fraction.
The one caveat is directional.
The guarantee is built for properties that hold at \textsc{the success} terminals, where $P_D$ \emph{underestimates}.
For a \emph{violation} property such as \textsc{misordered}, this under-estimation runs the other way: a non-zero per-input verdict proves a violation exists, but a
zero verdict does not certify its absence.

\subsection{CEGAR-Style Refinement}\label{sec:cegar}

To tighten the upper bound of Theorem~\ref{thm:soundness}, we adapted the CEGAR
principle~\cite{Clarke00cegar} to the DTMC setting.
Each round ranks all pruned $(s,t)$ pairs by impact
$\Pr[\text{reach}\;s] \cdot p_{\text{prune}}(s,t)$,
re-expands the top-$K$ pairs, splices new subtrees into $D$, re-runs PRISM, and
halts when $P_D(\textsc{low\_prob}) \le \varepsilon_{\text{target}}$ or the
round budget is exhausted.

\subsection{Counterexample Extraction}\label{sec:witness-method}

For any violation, the pipeline extracts the falsifying
trace with the highest joint probability by running Dijkstra's
algorithm on the DTMC with $-\log p$ edge weights, then decodes the
state-ID path back to the SUT token vocabulary.

\section{Evaluation}

The evaluation is organized around four research questions that span
both case studies. We address RQ1--3 as part of case study 1. Case study 2 focuses on portability, while also offering a second perspective on RQ1.

\noindent\textbf{RQ1 (Specification informativeness)}: Does the pipeline expose failure modes that test-set accuracy cannot report (Secs.~\ref{sec:rq1} and~\ref{sec:rq4-oracle})?

\noindent\textbf{RQ2 (Population coverage)}: Can per-input PRISM verdicts be aggregated into a population-level deployment guarantee (Sec.~\ref{sec:rq2})?

\noindent\textbf{RQ3 (Counterexamples and adaptive refinement)}: Does the pipeline produce inspectable counterexamples for violated specifications, and does CEGAR-style refinement tighten the soundness interval (Sec.~\ref{sec:rq3})?

\noindent\textbf{RQ4 (Portability)}: Does the methodology generalize to an SUT with a substantially larger vocabulary and a semantic oracle external to the model (Sec.~\ref{sec:smiles})?

\subsection{Case Study 1: CAPP}\label{sec:capp-eval}

\subsubsection{System Under Test}\label{sec:capp-sut}

The SUT is the trained GPT-2-based CAPP model of Stathatos et al.~\cite{STATHATOS2026103233}. 
We treat it as a black box and change neither its architecture nor its training.
The CAPP task maps a six-feature discrete part description (geometry, holes, threads, surface finish, tolerance, batch size) to an ordered sequence of manufacturing operations.
The input space is fully enumerable: 7{,}840 unique parts arise from the Cartesian product of the feature domains.
Outputs are sequences of process tokens delimited by \texttt{<EOC>} and terminated by \texttt{<EOS>}, encoding one or more alternative process chains.
The architecture is a 4-layer GPT-2 with 64-dimensional embeddings and a 53-token vocabulary (29 feature tokens, 20 process tokens, 4 special tokens).
The authors released eight checkpoints trained on dataset fractions from 1\% to 70\%; the model reaches 100\% sequence accuracy on the held-out test set from 15\% training onward.

As an additional domain specification, we consider the phase ordering of manufacturing process steps.
The 20 process tokens partition into three manufacturing phases:
\emph{Primary} (7 tokens, shape-giving), \emph{Secondary} (9 tokens,
subtractive refinement), and \emph{Finishing} (4 tokens, surface quality).
By manufacturing convention, operations within each chain must follow the
partial order $\text{Primary} \to \text{Secondary}^* \to \text{Finishing}^*$:
exactly one Primary operation, then zero or more Secondary, then zero or more Finishing, without backward phase transition.
This constraint is external to the neural model, which is not
trained to enforce it explicitly, and is the domain-specific
property of interest for the CAPP case study, encoded as the
\texttt{ordered}/\texttt{misordered} PCTL labels of
Table~\ref{tab:properties}.

We evaluated all eight training fractions (Table~\ref{tab:training-dtmc}, omitting the fixed 1{,}176-part validation split), building a per-input DTMC for every test sample and verifying seven of the eight specifications of Table~\ref{tab:properties} with PRISM (\textsc{critical} is obtained by forward propagation), yielding 37{,}554 DTMCs and 262{,}878 checking runs with zero solver failures.
All experiments ran on an Apple M2 Max (32\,GB RAM) using PRISM 4.10 with the explicit-state engine and four-way parallelism.

\begin{table}[t]
  \centering
  \caption{Training configurations with accuracy~\cite{STATHATOS2026103233}
    and DTMC structural statistics.}
  \label{tab:training-dtmc}
  \small
  \begin{tabular}{@{}rrrrr||rrrr@{}}
    \toprule
    & \multicolumn{4}{c}{Training} & \multicolumn{4}{c}{DTMC structure} \\
    \cmidrule(lr){2-5} \cmidrule(lr){6-9}
    Config & Train & Test & Tkn acc. & Seq.\ acc. &
             Avg states & Max states & Avg paths & Multi-branch \\
    \midrule
    1\%  & 78      & 6{,}586 & 93.5\%  & 47.8\%  & 933.6 & 2{,}208 & 64.7 & 100\%  \\
    5\%  & 392     & 6{,}272 & 99.7\%  & 96.1\%  & 55.4  & 560     & 5.2  & 82\%   \\
    10\% & 784     & 5{,}880 & 100.0\% & 99.7\%  & 16.3  & 54      & 1.0  & 1.8\%  \\
    15\% & 1{,}176 & 5{,}488 & 100.0\% & 100.0\% & 16.2  & 31      & 1.0  & 0.6\%  \\
    20\% & 1{,}568 & 5{,}096 & 100.0\% & 100.0\% & 16.3  & 37      & 1.0  & 1.5\%  \\
    30\% & 2{,}352 & 4{,}312 & 100.0\% & 100.0\% & 16.2  & 23      & 1.0  & 0.05\% \\
    50\% & 3{,}920 & 2{,}744 & 100.0\% & 100.0\% & 16.3  & 23      & 1.0  & 0.04\% \\
    70\% & 5{,}488 & 1{,}176 & 100.0\% & 100.0\% & 18.0  & 37      & 1.0  & 0\%    \\
    \bottomrule
  \end{tabular}
\end{table}

\subsubsection{RQ1: Failure Modes Invisible to Accuracy}\label{sec:rq1}

We answer along two dimensions: \emph{probability-mass divergence}
between accuracy and the specification
(Table~\ref{tab:mass}) and \emph{structural confidence diagnostics}
derived from the DTMC (Table~\ref{tab:training-dtmc}).

Accuracy and verification verdicts diverge.
From 15\% training onward, test-set sequence accuracy is 100\%, yet mean
$P(\textsc{success})$ ranges from 0.907 (15\%) to 0.974 (70\%).
By Theorem~\ref{thm:soundness}, the 70\% model's true reachability is
certified to lie in $[0.974, 1.000]$, with the gap entirely attributable
to \textsc{low\_prob}.
The pipeline separates two failure modes that accuracy conflates: exact
ground-truth match (what accuracy measures) and concentration of
probability mass on \emph{any} valid completion.

At the same time, structural validity is essentially preserved.
$P(\textsc{invalid}) = 0$ in every configuration except 1\% (0.009).
Even under severe data starvation, the model has implicitly learned the
syntactic well-formedness of the output language; this is itself a
verified property, not a stochastic observation.

$P(\textsc{success})$ flattens from 30\% training onward, gaining 0.019 between 30\% and 50\% and a further 0.006 between 50\% and 70\% (0.949, 0.968, 0.974).
A per-token threshold sweep on the 70\% model (200-input subsample) decomposes the remaining gap into two parts with different origins.
The first is a threshold artifact, and it is small.
Across $\tau \in \{0.5, 0.25, 0.1\}\%$ the value barely moves ($P(\textsc{success}) = 0.9746 \!\to\! 0.9747$), and only at $\tau = 0.01\%$ does appreciable mass re-enter \textsc{success} ($0.981$, with \textsc{low\_prob} falling from $2.5\%$ to $1.9\%$) before saturating again.

The second part is irreducible.
A residual $\approx\!1.9\%$ survives even at $\tau = 10^{-4}$ and cumulative floor $\rho = 10^{-6}$.
This residual is not an artifact of pruning, but a genuine sub-threshold \emph{dispersion}.
The model spreads probability across valid \emph{alternative} completions, each individually below $\tau$.
This dispersion is steep only when training data is scarce.
At 1\% training, just 36\% of the success mass concentrates on the canonical target, rising to 79\% at 5\% and above 99.5\% from 10\% onward.
For a deployed model the dispersion is therefore negligible.
Yet, it is invisible to both greedy decoding and test accuracy, which see only the canonical target whether it holds 36\% or 99.5\% of the mass.

Finally, the DTMC structure is a confidence diagnostic.
Table~\ref{tab:training-dtmc} reports structural statistics: the collapse from 934 mean states (1\%) to 18 (70\%), with the multi-branch fraction falling from 100\% to 0\%, traces the transition from a diffuse model to a near-deterministic one.
The critical-state diagnostic sharpens the picture: 
at 1\% training,
97.9\% of inputs have at least one critical state (mean 14.2, max 107); 
at 5\% this drops to 5.5\%; 
from 15\% onward, none is observed.
The pipeline thus certifies not only that the 70\% model is accurate but that it never approaches a decision boundary.
RQ1 is addressed again in the SMILES case study (Sec.~\ref{sec:rq4-oracle}) for which the
structural specification is \emph{not} sufficient.

\begin{table}[t]
  \centering
  \caption{Mean probability mass decomposition (Corollary~\ref{cor:mass}; max
    per-input deviation $<10^{-10}$ across all 37{,}554 DTMCs).}
  \label{tab:mass}
  \small
  \begin{tabularx}{\textwidth}{@{}r*{5}{>{\raggedleft\arraybackslash}X}@{}}
    \toprule
    Config & $P(\textsc{succ.})$ & $P(\textsc{low\_p.})$ &
             $P(\textsc{inv.})$ & $P(\textsc{trunc.})$ & Sum \\
    \midrule
    1\%  & 0.128 & 0.863 & 0.009 & 0.000 & 1.000 \\
    5\%  & 0.629 & 0.371 & 0.000 & 0.000 & 1.000 \\
    10\% & 0.845 & 0.155 & 0.000 & 0.000 & 1.000 \\
    15\% & 0.907 & 0.093 & 0.000 & 0.000 & 1.000 \\
    20\% & 0.922 & 0.079 & 0.000 & 0.000 & 1.000 \\
    30\% & 0.949 & 0.051 & 0.000 & 0.000 & 1.000 \\
    50\% & 0.968 & 0.032 & 0.000 & 0.000 & 1.000 \\
    70\% & 0.974 & 0.026 & 0.000 & 0.000 & 1.000 \\
    \bottomrule
  \end{tabularx}
\end{table}

\subsubsection{RQ2: Population Coverage and Operating Envelope}\label{sec:rq2}

\paragraph{Per-input ordering verdicts.}
From 10\% training onwards, $P_D(\textsc{misordered}) = 0$ across the entire test set: every retained success path is then correctly ordered, with only 0.3\% of retained probability mass misordered even at 1\% training. 
Because \textsc{misordered} is a violation property, the zero verdict guarantees that no \emph{retained} path misorders, but does not certify that the pruned sub-threshold mass is free of misordering.
This ordering property is invisible to sequence-level accuracy, and the coverage layer turns it into a population claim.

\paragraph{Population coverage.}
Table~\ref{tab:pac} reports the population ordering coverage $\hat\mu(\theta)$ at $\theta = 0.90$.
A sharp transition occurs between 10\% and 30\% training: coverage climbs from $\hat\mu = 0.252$ at 10\% to $\hat\mu = 0.971$ at 30\%, the smallest fraction at which the population is covered at the deployment level $c = 0.90$.
From 50\% onwards, every test input meets the threshold ($\hat\mu = 1.000$).
Because the curve is built from DTMC lower bounds, Lemma~\ref{lem:conservative} makes it a conservative estimate of the SUT's true coverage.

\begin{table}[t]
  \centering
  \caption{Population ordering coverage $\hat\mu$ for
    $P(\textsc{ordered}) \ge 0.90$ across training fractions
    ($N$ = test-set size).}
  \label{tab:pac}
  \small
  \begin{tabularx}{\textwidth}{@{}l*{8}{>{\raggedleft\arraybackslash}X}@{}}
    \toprule
    Config     & 1\%   & 5\%   & 10\%  & 15\%  & 20\%  & 30\%           & 50\%  & 70\% \\
    \midrule
    $N$        & 6{,}586 & 6{,}272 & 5{,}880 & 5{,}488 & 5{,}096 & 4{,}312 & 2{,}744 & 1{,}176 \\
    $\hat\mu$  & 0.000 & 0.004 & 0.252 & 0.576 & 0.705 & \textbf{0.971} & 1.000 & 1.000 \\
    \bottomrule
  \end{tabularx}
\end{table}

\paragraph{Operating envelope over decoder temperature.}
A single-temperature coverage value answers a single what-if question; a sweep over decoder temperatures characterises a \emph{safe operating envelope}.
We sweep $T \in \{0.5, 0.7, 1.0, 1.3, 2.0\}$ for the 1\%, 10\%, 30\%, and
70\% configurations and re-run the full PCTL suite at each temperature.
For the 30\% and 70\% models coverage stays at $\hat\mu \ge 0.97$ for $T \le 1.0$ and drops sharply at $T = 1.3$.
The 1\% model never reaches $\hat\mu = 0.90$; even at $T = 0.5$ its coverage is $0.789$.
The certified safe envelope is therefore $T \le 1.0$: above this the increased entropy places enough mass below $\tau$ to break the ordering coverage.

\paragraph{Best-of-$N$ as a certifiable decoder.}
The same propagated terminal probabilities also score, in closed form and
with no extra model runs, the success rate of drawing $N$ samples and
keeping a verifier-accepted one \cite{cobbe2021verifiers}, $\mathrm{pass}@N = \mathbb{E}_x[\,1-(1-P_x(\phi))^{N}]$.
At $T = 1.0$ best-of-$N$ never overtakes greedy decoding in the 1\% regime (the
sampling mass is too diffuse), unless we sharpen the decoder (Table~\ref{tab:bestofn-temp}).
At $T = 0.5$, best-of-$10$ makes every input yield a verifier-\textsc{ordered}
completion within ten draws ($\mathrm{pass}@10(\textsc{ordered}) = 1.00$) and
surfaces $1.87$ distinct valid plans per input against greedy's one.
Since ordering-validity, unlike correctness, is checkable at run time, the verifier doubles as selector, so this \textsc{ordered} result is the deployable claim; the \textsc{correct} column ($0.56$ at $T = 0.5$ against greedy's $0.47$) is an oracle upper bound on how much an ideal selector could additionally recover.
This makes best-of-$N$ at $T < 1.0$ the recommended decoder in the data-starved regime, where greedy decoding is weakest, i.e., the 1\% checkpoint is a deliberate worst case and not a deployment target. 
For well-trained (${\ge}15\%$) models, greedy decoding is already near-deterministic and the decoding gap closes, so there the pipeline's value shifts from recovering plans to certifying that little hidden mass remains.

\begin{table}[t]
  \centering
  \caption{Best-of-$N$ for the 1\% model ($N{=}10$); greedy correctness is $0.47$ at every $T$. $\mathrm{pass}@10$ on \textsc{correct} (oracle) and \textsc{ordered} (deployable), and distinct valid plans.}
  \label{tab:bestofn-temp}
  \small
  \begin{tabularx}{\textwidth}{@{}r*{3}{>{\raggedleft\arraybackslash}X}@{}}
    \toprule
    $T$ & $\mathrm{pass}@10(\textsc{correct})$ & $\mathrm{pass}@10(\textsc{ordered})$ & coverage$@10$ \\
    \midrule
    0.5 & 0.56 & 1.00 & 1.87 \\
    0.7 & 0.54 & 1.00 & 2.20 \\
    1.0 & 0.20 & 0.52 & 0.87 \\
    1.3 & 0.05 & 0.38 & 0.49 \\
    2.0 & 0.00 & 0.34 & 0.39 \\
    \bottomrule
  \end{tabularx}
\end{table}

\subsubsection{RQ3: Counterexamples and Adaptive Refinement}\label{sec:rq3}

The maximum-likelihood counterexample algorithm turns a quantitative
verdict into an inspectable falsifying trace.
We extract the most likely misordering trace from the DTMC of the
1\%-model input with the highest misordering probability.
The decoded trace (Figure~\ref{fig:witness}) carries joint probability
$3.5 \times 10^{-3}$.
The model emits a Finishing-phase operation before any Primary
operation, violating the phase-ordering property.
This trace is never selected by greedy decoding and is invisible to
test-set evaluation.
The counterexample algorithm answers not just \emph{whether} a
violation exists, but \emph{which} trace is most likely, and reports its formal probability. 
The result is a per-input artifact that a deployment engineer can inspect.

\begin{figure}[t]
  \centering
  \resizebox{0.7\textwidth}{!}{%
  \begin{tikzpicture}[
    tok/.style={draw, rounded corners=2pt, fill=gray!10,
                minimum height=0.55cm, inner sep=4pt,
                font=\small\ttfamily},
    viol/.style={tok, fill=red!20, thick},
    plab/.style={midway, above, font=\scriptsize},
    arr/.style={-Stealth, thick},
  ]
  \node[tok] (init) {init};
  \node[viol, right=0.6cm of init] (g)  {[F] 5-axis Grinding};
  \node[tok, right=0.6cm of g]     (e1) {<EOC>};
  \node[tok, right=0.6cm of e1]    (sc) {[P] Sand Casting};
  \node[tok, below=0.95cm of init.west, anchor=west] (e2) {<EOC>};
  \node[tok, right=0.6cm of e2]    (ic) {[P] Investment Casting};
  \node[tok, right=0.6cm of ic]    (e3) {<EOC>};
  \node[tok, right=0.6cm of e3]    (eos) {<EOS>};
  \draw[arr] (init) -- node[plab] {.03} (g);
  \draw[arr] (g)    -- node[plab] {.85} (e1);
  \draw[arr] (e1)   -- node[plab] {.71} (sc);
  \draw[arr] (sc.south) -- ++(0,-0.3)
    -| node[pos=0.25, above, font=\scriptsize] {.63} (e2.north);
  \draw[arr] (e2)   -- node[plab] {.45} (ic);
  \draw[arr] (ic)   -- node[plab] {.80} (e3);
  \draw[arr] (e3)   -- node[plab] {.76} (eos);
  \end{tikzpicture}}
  \caption{Maximum-likelihood misordering counterexample for the 1\% model.
    Edge labels are (rounded) conditional next-token probabilities.
    The first operation is \emph{5-axis Grinding} (Finishing
    phase, highlighted), before any Primary operation.}
  \label{fig:witness}
\end{figure}
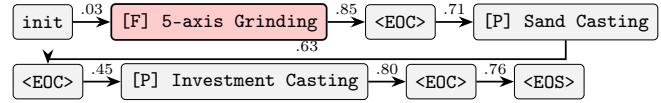

The CEGAR loop then tightens the soundness interval adaptively.
Here, we run the CEGAR loop on $N = 50$ inputs per configuration (top-$K=5$, $\varepsilon_{\text{target}}=10^{-3}$) and report the distribution of $P_D(\textsc{low\_prob})$ at rounds 0 and 5 (Table~\ref{tab:refinement_dist}).
Two patterns are robust.
First, per-sample dispersion at round~0 grows as training fraction shrinks, and refinement contracts the inter-quartile range everywhere.
Second, configurations whose round-0 median is large (1\% and 5\%) show the smallest relative improvement: five rounds cut the 1\% median by only a few percent, while a 30\%-model sample benefits by 25--35\%.
Refinement is therefore most \emph{effective} on the moderate configurations, even though the certified interval is \emph{widest} on the data-starved ones: there much of the gap is the irreducible sub-threshold dispersion of Sec.~\ref{sec:rq1}, which bounded top-$K$ re-expansion cannot reclaim. 
Closing it would need a lower threshold $\tau$ (with diminishing returns, Sec.~\ref{sec:rq1}) rather than more refinement rounds.

\begin{table}[t]
\centering
\caption{CEGAR refinement: distribution of $P_D(\textsc{low\_prob})$ across
  $N=50$ sampled inputs at the initial DTMC (round 0) and after 5 refinement
  rounds.  Quartiles are taken over the per-sample verification results.}
\label{tab:refinement_dist}
\small
\begin{tabularx}{\textwidth}{@{}r*{6}{>{\raggedleft\arraybackslash}X}@{}}
\toprule
 & \multicolumn{3}{c}{Round 0} & \multicolumn{3}{c}{Round 5} \\
\cmidrule(lr){2-4} \cmidrule(lr){5-7}
Config & $P_{25}$ & Median & $P_{75}$ & $P_{25}$ & Median & $P_{75}$ \\
\midrule
  1\%  & 0.791 & 0.929 & 0.975 & 0.761 & 0.919 & 0.961 \\
  5\%  & 0.312 & 0.396 & 0.516 & 0.235 & 0.313 & 0.431 \\
  10\% & 0.127 & 0.163 & 0.219 & 0.090 & 0.120 & 0.174 \\
  15\% & 0.057 & 0.090 & 0.121 & 0.039 & 0.067 & 0.093 \\
  20\% & 0.054 & 0.076 & 0.102 & 0.037 & 0.055 & 0.079 \\
  30\% & 0.024 & 0.048 & 0.066 & 0.016 & 0.035 & 0.051 \\
  50\% & 0.024 & 0.031 & 0.039 & 0.018 & 0.023 & 0.030 \\
  70\% & 0.019 & 0.027 & 0.033 & 0.014 & 0.021 & 0.026 \\
\bottomrule
\end{tabularx}
\end{table}

\subsection{Case Study 2: SMILES Molecular Generation}\label{sec:smiles}

Case study 1 exercises the pipeline on a small, fully enumerable input space with a domain specification (phase ordering) defined over an output grammar the SUT already enforces almost perfectly.
 The question is whether the extraction of the DTMC or the specification layer depends on these favorable conditions.
RQ4 asks whether the methodology generalises to a SUT that stresses both: a much larger vocabulary that produces much larger DTMCs, and a \emph{semantic} specification (chemical validity) that the output grammar does \emph{not} enforce by construction.
We answer in two parts: an \emph{operational portability} assessment, and a \emph{specification-informativeness} assessment showing that the semantic specification surfaces failures invisible both to a structural specification and to accuracy.

\subsubsection{System Under Test}\label{sec:smiles-sut}

The SUT is \texttt{gpt2\_zinc\_87m}\footnote{Model: \url{https://huggingface.co/entropy/gpt2\_zinc\_87m}}, an 87\,M-parameter GPT-2 pretrained on the ZINC molecular database with a $\approx 2{,}700$-token byte-pair encoding (BPE), which is roughly $50\times$ the CAPP vocabulary.
Test inputs are 200 short prefix strings in the SMILES notation (simplified molecular-input line-entry system)~\cite{DBLP:journals/jcisd/Weininger88}, drawn from the ZINC alphabet (single atoms and two-atom fragments, e.g., \texttt{C}, \texttt{N}, \texttt{CC}, \texttt{c1}); the pipeline builds one DTMC per prompt.
Extraction parameters are unchanged ($\tau = 0.5\%$, $\rho = 10^{-4}$, $d_{\max} = 20$).

As a domain specification, we consider the chemical validity as assessed by RDKit's MolFromSmiles.
That is, RDKit is our \emph{test oracle}~\cite{Barr15}: a black-box pass/fail verdict procedure, independent of the SUT, that the pipeline consumes as a PCTL label.
This oracle enforces semantic constraints (valence, ring-closure consistency, atom-count balance) that are \emph{not} part of the SMILES grammar, so structural completion and chemical validity are decoupled. 
This gap is known for neural SMILES generators, which can emit syntactically complete but chemically invalid strings~\cite{krenn2020selfies}.
A \texttt{P=?[F valid\_smiles]} label replaces the CAPP \texttt{ordered}/\texttt{misordered} pair in Table~\ref{tab:properties}. 
No other pipeline component changes.

Against CAPP's 53 tokens and 16-state DTMCs, SMILES uses $\approx\!2{,}700$ BPE tokens and produces DTMCs of mean 13{,}150 states (max 37{,}397), with mean extraction 485\,s and PRISM 43\,s per prompt vs.\ sub-second for CAPP.

\subsubsection{RQ4: Operational Portability}\label{sec:rq4}

All 200 PRISM instances completed without solver failure, matching the
0\% error rate on CAPP.
The extracted DTMCs are substantially larger than CAPP's: mean 13{,}150 states
(p25 = 10{,}413; median = 12{,}408; p75 = 16{,}287; max = 37{,}397), roughly
$800\times$ the size of well-trained CAPP DTMCs and $14\times$ the size of the
worst-case 1\%-trained CAPP DTMCs.
Build time is dominated by extraction (mean 485\,s, max 1{,}364\,s) rather than verification (mean 43\,s, max 370\,s); both scale super-linearly with vocabulary size but remain tractable on a laptop-class machine.
Still, a full sweep over a ZINC-sized prompt set is beyond the scope of this work.

The pipeline is therefore vocabulary-agnostic:
no component assumes a fixed token set or a small state space, and the soundness theorem and CEGAR loop carry over unchanged.
Scale also makes the depth limit binding for a non-trivial fraction of prompts: mean $P_D(\textsc{truncated}) = 0.018$ across the 200 inputs, whereas on CAPP it was $\le 10^{-3}$ throughout (Table~\ref{tab:mass}).
By Theorem~\ref{thm:soundness} this widens the certified soundness interval. The gap here is dominated by the \textsc{truncated} term rather than the diffuse \textsc{low\_prob} dispersion of case study~1, so it is the component most reducible by CEGAR or a larger $d_{\max}$.

\subsubsection{RQ4+RQ1: Specification Informativeness}\label{sec:rq4-oracle}

The RDKit specification exposes a failure class that does not have an analog in the CAPP study.
Among the 90 prompts for which at least one chemically valid terminal is produced, mean $P(\textsc{valid\_smiles}) = 0.084$ while mean $P(\textsc{success}) = 0.248$.
Conditional on a structurally complete terminal, roughly $66\%$ ($1 - 0.084/0.248$) of probability mass corresponds to invalid molecules.
At the prompt level, 181 of 200 prompts produce at least one chemically invalid but structurally complete terminal; 87 produce both valid and invalid terminals; 94 produce only invalid structurally complete terminals; and a further 16 produce no structurally complete terminal at all. 
For 110 prompts in total, no chemically valid SMILES is ever produced.

This is the structural-vs-semantic validity gap that the CAPP study
cannot illustrate.
In CAPP, $P(\textsc{invalid}) \approx 0$ from 5\% training onwards:
the token grammar happens to encode well-formedness, so the
structural-validity check of Sec.~\ref{sec:extraction} carries most of
the load.
In SMILES, structural completion is necessary but not sufficient for
semantic validity, and the model reliably misallocates confidence to
chemically meaningless strings.
Test-set accuracy, i.e., exact match against a reference, would
not distinguish these cases at all.
$P(\textsc{success})$ alone would give a misleadingly optimistic
estimate; the gap between $P(\textsc{success})$ and $P(\textsc{valid\_smiles})$ is the formal measure of the discrepancy.

From a methodology viewpoint, the RDKit check is integrated as an external
test oracle: integrating it required no change to extraction, PRISM
export, or coverage aggregation.
This shows the oracle-independence of the methodology in one case, supporting the more general claim of Sec.~\ref{sec:intro}: any executable
oracle that maps a terminal sequence to a pass/fail verdict can serve as
the domain specification, and the same three artifacts (DTMC
abstraction, PCTL specification, population coverage) serve different
domains with no methodology rework.

\section{Threats to Validity}\label{sec:threats}

\emph{Internal.}
The DTMC under-approximates the SUT, and the soundness interval widens
monotonically with $\tau$ and contracts with $\rho$.
The CEGAR loop mitigates this by tightening the interval to a target
$\varepsilon$.
The critical-state gap threshold (10\%) is used only as a diagnostic flag
and does not enter the soundness theorem; we verified that critical-state
rankings are qualitatively stable under gap thresholds in $\{5\%, 10\%,
20\%\}$ on the only configuration with non-trivial critical-state mass
(1\% training).

\emph{Construct.}
The CAPP ordering specification encodes $P \to S^* \to F^*$ and does
not capture richer manufacturing constraints (resource contention,
tool-change costs, batch-dependent process choice).
The SMILES specification is RDKit's MolFromSmiles, which is
the de facto chemical-validity check in the community but is not,
strictly speaking, a formal specification.
The population coverage assumes that the test split is representative of the
deployment distribution; for CAPP the input space is finite and fully
enumerable, so coverage on the held-out set is an exact enumerable
fraction for any random split.
For SMILES, the 200 prompts used here are not a random sample of any
deployment distribution, and the coverage is read accordingly: it
characterises the gap between $P(\textsc{success})$ and the
specification, not a population coverage guarantee, and we therefore
report raw $P(\textsc{success})$ and $P(\textsc{valid\_smiles})$ without
a population layer.

\emph{External.}
The CAPP and SMILES studies span very different vocabularies (53 vs.\ $\approx 2{,}700$ tokens), DTMC sizes ($800\times$), and specification types (syntactic--temporal vs.\ semantic--external).
This supports a portability claim, i.e., the same pipeline runs on both with no source-code change, but not a generality claim across all autoregressive transformers.
We have not yet exercised long-horizon generation ($d \gg 20$), non-GPT
architectures, or diverse sampling schemes.

\section{Conclusion}\label{sec:conclusion}

We have presented a probabilistic model-checking pipeline for autoregressive transformers that combines a DTMC abstraction with a soundness theorem, exact PCTL verification, and population coverage over the input space. 
Where test-set accuracy scores a single greedy trajectory against a reference, it is silent on the probability mass a model places elsewhere and on domain constraints the reference never exercises. 
Our pipeline is both visible and quantifiable.
It recasts the deployment question as whether the gap between a model's confident output and a domain specification is acceptable, and reports a population coverage curve for it.
Its natural place is at development and model-selection time: certified coverage and best-of-$N$ analysis inform the decoder and the training-data budget before deployment, rather than monitoring a fixed model at run time.
Two properties make our approach practical. 
The specification is supplied by the application rather than the method: a temporal ordering property and an external executable oracle plug into the same three artifacts with no pipeline change, so the methodology is indifferent to domain. 
The soundness guarantee is conservative by construction, so the reported coverage can be trusted to understate rather than overstate the truth.
We evaluate this on two contrasting systems, differing by an order of magnitude in vocabulary and by specification type.
The results give evidence of portability, though not yet generality across all architectures and decoding schemes.
Future work includes longer-horizon generation, richer property classes, and sampling strategies beyond temperature scaling. 

\begin{credits}
\subsubsection{\ackname}
This work is funded by the European Union under grant agreement number 101091783 (MARS Project).

\subsubsection{\discintname}
The authors declare that they have no competing financial interests.
\end{credits}

\bibliographystyle{splncs04}
\bibliography{refs}

\begin{thebibliography}{10}
\providecommand{\url}[1]{\texttt{#1}}
\providecommand{\urlprefix}{URL }
\providecommand{\doi}[1]{https://doi.org/#1}

\bibitem{Angelopoulos23}
Angelopoulos, A.N., Bates, S., Fisch, A., Lei, L., Schuster, T.: Conformal risk
  control. In: International Conference on Learning Representations ({ICLR})
  (2024)

\bibitem{bacci2020probabilistic}
Bacci, E., Parker, D.: Probabilistic guarantees for safe deep reinforcement
  learning. In: International Conference on Formal Modeling and Analysis of
  Timed Systems. pp. 231--248. Springer (2020).
  \doi{10.1007/978-3-030-57628-8_14}

\bibitem{Barr15}
Barr, E.T., Harman, M., McMinn, P., Shahbaz, M., Yoo, S.: The oracle problem in
  software testing: {A} survey. IEEE Transactions on Software Engineering
  \textbf{41}(5),  507--525 (2015). \doi{10.1109/TSE.2014.2372785}

\bibitem{Clarke00cegar}
Clarke, E., Grumberg, O., Jha, S., Lu, Y., Veith, H.: Counterexample-guided
  abstraction refinement. In: Proc.\ 12th Int.\ Conf.\ Computer Aided
  Verification (CAV). LNCS, vol.~1855, pp. 154--169. Springer (2000).
  \doi{10.1007/10722167_15}

\bibitem{cobbe2021verifiers}
Cobbe, K., Kosaraju, V., Bavarian, M., Chen, M., Jun, H., Kaiser, L., Plappert,
  M., Tworek, J., Hilton, J., Nakano, R., Hesse, C., Schulman, J.: Training
  verifiers to solve math word problems. arXiv preprint arXiv:2110.14168
  (2021). \doi{10.48550/arXiv.2110.14168}

\bibitem{CousotCousot77}
Cousot, P., Cousot, R.: Abstract interpretation: A unified lattice model for
  static analysis of programs by construction or approximation of fixpoints.
  In: Proc.\ 4th {ACM} Symp.\ Principles of Programming Languages ({POPL}). pp.
  238--252. ACM (1977). \doi{10.1145/512950.512973}

\bibitem{David15}
David, A., Larsen, K.G., Legay, A., Miku\v{c}ionis, M., Poulsen, D.B.: Uppaal
  {SMC} tutorial. International Journal on Software Tools for Technology
  Transfer  \textbf{17},  397--415 (2015). \doi{10.1007/s10009-014-0361-y}

\bibitem{GordonHNR14}
Gordon, A.D., Henzinger, T.A., Nori, A.V., Rajamani, S.K.: Probabilistic
  programming. In: Proc.\ Future of Software Engineering (FOSE). pp. 167--181.
  ACM (2014). \doi{10.1145/2593882.2593900}

\bibitem{DBLP:conf/esann/GrossS24}
Gross, D., Spieker, H.: Safety-oriented pruning and interpretation of
  reinforcement learning policies. In: 32nd European Symposium on Artificial
  Neural Networks, Computational Intelligence and Machine Learning ({ESANN})
  (2024). \doi{10.14428/esann/2024.ES2024-71}

\bibitem{DBLP:conf/sac/GrossS25}
Gross, D., Spieker, H.: {PCTL} model checking for temporal {RL} policy safety
  explanations. In: Proceedings of the 40th {ACM/SIGAPP} Symposium on Applied
  Computing ({SAC}). pp. 1514--1521. {ACM} (2025).
  \doi{10.1145/3672608.3707759}

\bibitem{DBLP:conf/ictai/GrossSG25}
Gross, D., Spieker, H., Gotlieb, A.: Bounded {PCTL} model checking of large
  language model outputs. In: 37th {IEEE} International Conference on Tools
  with Artificial Intelligence ({ICTAI}). pp. 106--113. {IEEE} (2025).
  \doi{10.1109/ICTAI66417.2025.00022}

\bibitem{DBLP:journals/fac/HanssonJ94}
Hansson, H., Jonsson, B.: A logic for reasoning about time and reliability.
  Formal Aspects Comput.  \textbf{6}(5),  512--535 (1994).
  \doi{10.1007/BF01211866}

\bibitem{holtzman2020curious}
Holtzman, A., Buys, J., Du, L., Forbes, M., Choi, Y.: The curious case of
  neural text degeneration. In: International Conference on Learning
  Representations ({ICLR}) (2020)

\bibitem{ji2023hallucination}
Ji, Z., Lee, N., Frieske, R., Yu, T., Su, D., Xu, Y., Ishii, E., Bang, Y.J.,
  Madotto, A., Fung, P.: Survey of hallucination in natural language
  generation. ACM Computing Surveys  \textbf{55}(12),  1--38 (2023).
  \doi{10.1145/3571730}

\bibitem{Katz17}
Katz, G., Barrett, C., Dill, D.L., Julian, K., Kochenderfer, M.J.: Reluplex: An
  efficient {SMT} solver for verifying deep neural networks. In: Proc.\ 29th
  Int.\ Conf.\ Computer Aided Verification (CAV). LNCS, vol. 10426 (2017)

\bibitem{Katz19}
Katz, G., Huang, D.A., Ibeling, D., Julian, K., Lazarus, C., Lim, R., Shah, P.,
  Thakoor, S., Wu, H., Zeljic, A., Dill, D.L., Kochenderfer, M.J., Barrett, C.:
  The marabou framework for verification and analysis of deep neural networks.
  In: Proc.\ 31st Int.\ Conf.\ Computer Aided Verification (CAV). LNCS, vol.
  11561 (2019)

\bibitem{krenn2020selfies}
Krenn, M., H{\"{a}}se, F., Nigam, A., Friederich, P., Aspuru{-}Guzik, A.:
  Self-referencing embedded strings {(SELFIES):} {A} 100{\%} robust molecular
  string representation. Machine Learning: Science and Technology
  \textbf{1}(4),  45024 (2020). \doi{10.1088/2632-2153/aba947}

\bibitem{Kwiatkowska11}
Kwiatkowska, M., Norman, G., Parker, D.: {PRISM} 4.0: Verification of
  probabilistic real-time systems. In: Proc.\ 23rd Int.\ Conf.\ Computer Aided
  Verification (CAV). LNCS, vol.~6806 (2011).
  \doi{10.1007/978-3-642-22110-1_47}

\bibitem{Legay10}
Legay, A., Delahaye, B., Bensalem, S.: Statistical model checking: An overview.
  In: Proc.\ Runtime Verification (RV). LNCS, vol.~6418, pp. 122--135. Springer
  (2010). \doi{10.1007/978-3-642-16612-9_11}

\bibitem{liang2023holistic}
Liang, P., Bommasani, R., Lee, T., {et al.}: Holistic evaluation of language
  models. Transactions on Machine Learning Research (TMLR)  (2023)

\bibitem{Ma2018}
Ma, L., Juefei-Xu, F., Zhang, F., Sun, J., Xue, M., Li, B., Chen, C., Su, T.,
  Li, L., Liu, Y., Zhao, J., Wang, Y.: {DeepGauge}: Multi-granularity testing
  criteria for deep learning systems. In: Proc.\ 33rd {IEEE/ACM} Int.\ Conf.\
  Automated Software Engineering ({ASE}). pp. 120--131. ACM (2018).
  \doi{10.1145/3238147.3238202}

\bibitem{MohriHashimoto2024}
Mohri, C., Hashimoto, T.: Language models with conformal factuality guarantees.
  In: Forty-first International Conference on Machine Learning ({ICML}) (2024)

\bibitem{Pei2017}
Pei, K., Cao, Y., Yang, J., Jana, S.: {DeepXplore}: Automated whitebox testing
  of deep learning systems. In: Proc.\ 26th {ACM} Symp.\ Operating Systems
  Principles ({SOSP}). pp. 1--18. ACM (2017). \doi{10.1145/3132747.3132785}

\bibitem{Quach2024}
Quach, V., Fisch, A., Schuster, T., Yala, A., Sohn, J.H., Jaakkola, T.S.,
  Barzilay, R.: Conformal language modeling. In: The Twelfth International
  Conference on Learning Representations ({ICLR}) (2024)

\bibitem{radford2019language}
Radford, A., Wu, J., Child, R., Luan, D., Amodei, D., Sutskever, I., et~al.:
  Language models are unsupervised multitask learners. OpenAI blog
  \textbf{1}(8), ~9 (2019)

\bibitem{Singh19}
Singh, G., Gehr, T., P\"{u}schel, M., Vechev, M.T.: An abstract domain for
  certifying neural networks. In: Proc.\ 46th {ACM} {SIGPLAN} Symp.\ Principles
  of Programming Languages ({POPL}). pp. 41:1--41:30. ACM (2019).
  \doi{10.1145/3290354}

\bibitem{STATHATOS2026103233}
Stathatos, E., Benardos, P., Vosniakos, G., Gross, D., Spieker, H., Gotlieb,
  A.: Large language models for high-level computer-aided process planning in a
  distributed manufacturing paradigm. Robotics and Computer-Integrated
  Manufacturing  \textbf{100},  103233 (2026). \doi{10.1016/j.rcim.2026.103233}

\bibitem{valmeekam2023planning}
Valmeekam, K., Marquez, M., Sreedharan, S., Kambhampati, S.: On the planning
  abilities of large language models - {A} critical investigation. In: Advances
  in Neural Information Processing Systems (NeurIPS). vol.~36 (2023)

\bibitem{Wang21}
Wang, S., Zhang, H., Xu, K., Lin, X., Jana, S., Hsieh, C., Kolter, J.Z.:
  Beta-crown: Efficient bound propagation with per-neuron split constraints for
  neural network robustness verification. In: Advances in Neural Information
  Processing Systems (NeurIPS) (2021)

\bibitem{DBLP:journals/jcisd/Weininger88}
Weininger, D.: Smiles, a chemical language and information system. 1.
  introduction to methodology and encoding rules. J. Chem. Inf. Comput. Sci.
  \textbf{28}(1) (1988). \doi{10.1021/CI00057A005}

\bibitem{Weiss2018}
Weiss, G., Goldberg, Y., Yahav, E.: Extracting automata from recurrent neural
  networks using queries and counterexamples. In: International Conference on
  Machine Learning ({ICML}). {PMLR} (2018)

\end{thebibliography}

\end{document}